\documentclass[aps,twocolumn,secnumarabic,nobalancelastpage,amsmath,amssymb,
nofootinbib]{revtex4}

\usepackage{amsmath,amsthm,amssymb,xypic,hyperref,amsfonts}
\usepackage{epsfig} 
\usepackage{color}
\usepackage{chapterbib}
\usepackage{graphics}      
\usepackage{graphicx}      
\usepackage{longtable}     
\usepackage{url}          
\usepackage{bm}            

\newcommand{\bra}[1]{\left\langle #1 \right|}
\newcommand{\ket}[1]{\left|#1\right\rangle}

\numberwithin{equation}{section}
\newtheorem{thm}[equation]{Theorem}

\newtheorem{prop}[equation]{Proposition}
\newtheorem{algo}[equation]{Algorithm}

\begin{document}
\title{Nonadditive Quantum Error Correcting Codes Adapted to the Amplitude Damping Channel}
\author         {Ruitian Lang}
\email          {percyl@mit.edu}
\author         {Peter W Shor}
\email          {shor@math.mit.edu}
\affiliation    {Department of Mathematics, MIT, Cambridge, MA02139}
\date{\today}

\begin{abstract}
A family of high rate quantum error correcting codes adapted to the amplitude damping channel is presented. These codes are nonadditive and exploit self-complementary structure to correct all first-order errors. Their rates can be higher than 1/2. The recovery operations of these codes can be generated by a simple algorithm and have a projection nature, which makes them potentially easy to implement.
\end{abstract}

\maketitle

\section{Introduction}
As quantum computation finds wide applications today, the difficulty of implementing quantum computers also arises from the expense of
qubits and their vulnerability to decoherence. Therefore, quantum error correction and fault tolerant quantum computation have been
extensively developed \cite{NC}. Among the most famous quantum error correcting codes are the [9,1] code \cite{91}, the CSS [7,1] code
\cite{71a,71b}, and the [5,1] code \cite{51a,51b}. The encoding rate of these codes are much lower than classical codes due to the
complicated nature of quantum decoherence. However, qubits are expensive, so it is desirable to have codes with higher rate. Unfortunately,
it has been shown that 1/5 (encoding 1 qubit into 5 qubits) is the highest rate for one-qubit code. In fact, quantum error correction codes
satisfy the following quantum Hamming bound \cite{NC}: if an $[n,k]$ code (encoding $k$ qubits into $n$ qubits) can correct errors at at most $t$ qubits and there are $a$ independent possible errors in one qubit, then

\begin{equation}
\sum_{j=0}^{t}\left( \begin{array}{c} n\\ j \end{array}\right)a^{j}2^{k}\leq 2^{n},
\label{hamming}
\end{equation}
It has been shown that if a code corrects all the errors in the Pauli group generated by the Pauli matrices operating on each qubit, then it can correct all types of errors \cite{NC}. Therefore, traditional quantum error correction has concentrated on the errors in the Pauli group, and hence $a=3$ in \eqref{hamming}. Within this framework, it is hard to improve the rate further.


However, not all the errors are equally likely in a realistic channel. If we assume some knowledge about the channel, such that only 1 or 2 independent errors are likely to occur in each qubit, then codes of higher rate can be obtained. Indeed, an iterative optimization method to investigate channel-adapted codes has been proposed by \cite{andrew} and \cite{iteration} and a family of $[2n+2,n]$ stabilizer codes, which correct all the first-order errors and some of the second-order ones in an amplitude damping channel, have been discovered by \cite{andrewthesis}. The amplitude damping channel is the tensor product of one-qubit amplitude damping channels, each of which consists of the following two operator elements:
\begin{equation}
E_{0}=\left( \begin{array}{cc}
1 & 0\\
0 & \sqrt{1-\gamma}\end{array}\right),
E_{1}=\left( \begin{array}{cc}
0 & \sqrt{\gamma}\\
0 & 0\end{array} \right),
\label{channel}
\end{equation}
where $\gamma$ is a parameter measuring the strength of the noise. Here $\ket{1}$ is pictured as an excited state and $\ket{0}$ the ground state. $\gamma$ is the probability of the transition from $\ket{1}$ to $\ket{0}$.


Another important concept, nonadditive codes, was introduced in recent research \cite{sfc, nonad1, nonad2} seeking to find codes with
higher rates than stabilizer codes. In these codes, the codewords do not form a subspace of $\mathbb{F}_{2}^{n}$ and thus the dimension of
the source space is not of the form $2^{k}$. In other words, these codes encode a fractional number of qubits. To avoid confusion, we
denote a (nonadditive) code which encodes a $k$ dimensional space into $n$ qubits by an $(n,k)$ code. This paper combines the ideas of
channel adaption and nonadditive codes to find a new family of codes adapted to the amplitude damping channel which outperform the
$[2n+2,n]$ stabilizer codes in \cite{andrewthesis}.

\section{Self-complementary Codes}
The family of codes to be presented exploits self-complementary structure, which is crucial for correcting errors in the amplitude damping channel. Self-complementary codes appear in different contexts \cite{sfc, andrewthesis} and play different roles. To understand the significance of the structure, we analyze how errors arise and are corrected as follows.


Let $\mathcal{C}$ be a quantum channel. The extent to which the channel preserves information is measured by the entanglement fidelity \cite{NC}
\begin{equation}
F(\rho,\mathcal{C}(\rho))=\mathrm{tr}\sqrt{\rho^{\frac{1}{2}}\mathcal{C}(\rho)\rho^{\frac{1}{2}}},
\label{fid}
\end{equation}
where $\rho$ is the input state. Sometimes, we are more interested in the worst case performance, which corresponds to the minimum entanglement fidelity
\begin{equation}
F_{\mathrm{min}}(\mathcal{C})=\min_{\ket{\psi}}F(\ket{\psi}\bra{\psi},\mathcal{C}(\ket{\psi}\bra{\psi})).
\end{equation}
The restriction to pure states is justified in \cite{NC}. In our case, $\mathcal{C}$ is composed of the encoding operation $\mathcal{U}$,
the amplitude damping channel $\mathcal{E}(\gamma)$ and the recovery operation $\mathcal{R}(\gamma)$. The dependence of $\mathcal{R}$ on
$\gamma$ suggests that the code be channel-adapted. The errors contained in $\mathcal{C}(\gamma)$, appropriately called the residual
errors, can be divided into two classes. First of all, an error may flip a codeword $\ket{w_{1}}\bra{w_{1}}$ to another codeword
$\ket{w_{2}}\bra{w_{2}}$, like an $X$ operator acting on this two-dimensional subspace. We call this an $X$-like error. Assuming no $X$-like errors (up to a specific order), we may still have an error that causes damping or a phase shift (or both) of an off-diagonal entry: $1-\mathrm{tr}(\ket{w_{2}}\bra{w_{1}}\mathcal{C}(\ket{w_{1}}\bra{w_{2}}))\neq 0$. Different errors cause this problem, including the $Z$ operator in the two dimensional space spanned by $\ket{w_{1}}$ and $\ket{w_{2}}$. For simplicity, we call these errors $Z$-like errors. The following proposition shows that this list is exhaustive, up to the first order.

\begin{prop}\label{XZerrors}
The code corrects all the first-order errors in the sense that $F_{\mathrm{min}}=1-o(\gamma^{2})$ if and only if $\mathcal{C}(\gamma)$ do not contain any $X$-like errors or $Z$-like errors to the first order.
\end{prop}
\begin{proof}
If $\mathcal{C}(\gamma)$ contains an $X$-like error on $\ket{w_{1}}$ or a $Z$-like error on $\ket{w_{1}}\bra{w_{2}}$, then choose $\ket{\psi}=\ket{w_{1}}$ or $\ket{\psi}=(\ket{w_{1}}+\ket{w_{2}})/\sqrt{2}$, respectively, to show that $F(\ket{\psi}\bra{\psi},\mathcal{C}(\ket{\psi}\bra{\psi}))$ contains a first-order term.

Conversely, assume that $\mathcal{C}(\gamma)$ does not contain any $X$-like or $Z$-like errors to the first order. Call the underlying Hilbert space $H$, which is finite dimensional. Then $\mathcal{C}(\gamma)$ is a linear operator on $\mathcal{L}(H,H)$, the space of the linear operators on $H$. By the assumption, the entries of the matrix $I-\mathcal{C}(\gamma)$ in the basis $\{\ket{w_{1}}\bra{w_{2}}\}$ do not contain first-order terms, so $\|I-\mathcal{C}(\gamma)\|=o(\gamma^{2})$. The conclusion follows from the continuity of $F(\rho,\sigma)$ and the compactness of the unit sphere in $H$.
\end{proof}


Consider an $(n,k)$ code. For a codeword $u\in\{0,1\}^{n}$, we denote by $\bar{u}$ its complement such that $\bar{u}_{j}=1-u_{j}$ for all $j$. We call a code self-complementary, if the code is spanned by kets of the form $(\ket{u}+\ket{\bar{u}})/\sqrt{2}$. We define the inner product of two codewords $u$ and $v$ by $u\cdot v=\sum_{j=1}^{n}u_{j}v_{j}$ and thus the Hamming weight of $u$ is $\|u\|^{2}$. By the definition \eqref{channel} of the amplitude damping channel, after passing through the channel, $(\ket{u}+\ket{\bar{u}})/\sqrt{2}$ becomes
\begin{equation}\label{fu'}
\ket{\tilde{f}(u)}=\frac{1}{\sqrt{2}}((1-\gamma)^{\frac{1}{2}\|u\|^{2}}\ket{u}+(1-\gamma)^{\frac{1}{2}\|\bar{u}\|^{2}}\ket{\bar{u}}).
\end{equation}
Let $\ket{f(u)}$ be the unit vector associated to $\ket{\tilde{f}(u)}$, and $\ket{g(u)}$ be the unit vector in $\mathrm{span\{\ket{u},\ket{\bar{u}}\}}$ orthogonal to $\ket{f(u)}$.
It is thus reasonable to correct $\ket{f(u)}$ to the codeword $(\ket{u}+\ket{\bar{u}})/\sqrt{2}$. When we talk about self-complementary codes adapted to the amplitude-damping channels, we will always assume this recovery operation. Now we will show that the self-complementary structure eliminates the $Z$-like errors provided that the $X$-like errors have been eliminated.


\begin{thm}\label{sc}
Let $\mathcal{C}$ be the composite channel of a self-complementary code subject to an amplitude damping channel. If $\mathcal{C}(\gamma)$ does not contain $X$-like errors to the first order, it does not contain $Z$-like errors either.
\end{thm}
\begin{proof}
Let $e^{(j)}\in\{0,1\}^{n}$ denote the error vector such that $e^{(j)}_{i}=\delta_{ij}$. By the assumption, if $\ket{w}=(\ket{u}+\ket{\bar{u}})/\sqrt{2}$ is in the code, then $\ket{u-e^{(j)}}$ will be corrected fully to $\ket{w}$ provided that $u_{j}=1$ since otherwise this results in a first-order $X$-like error. When $u_{j}=0$, we agree that $\ket{u-e^{(j)}}=0$. Now consider two codewords $\ket{w_{1}}=(\ket{u}+\ket{\bar{u}})/\sqrt{2}$ and $\ket{w_{2}}=(\ket{v}+\ket{\bar{v}})/\sqrt{2}$, where $u,v\in\{0,1\}^{n}$. We have
\begin{eqnarray*}
&&\mathcal{C}(\ket{w_{1}}\bra{w_{2}})\\
&=&\mathcal{R}[\ket{\tilde{f}(u)}\bra{\tilde{f}(v)}+\frac{\gamma}{2}\sum_{j=1}^{n}(\ket{u-e^{(j)}}\bra{v-e^{(j)}}\\
&&+\ket{u-e^{(j)}}\bra{\bar{v}-e^{(j)}}+\ket{\bar{u}-e^{(j)}}\bra{v-e^{(j)}}\\
&&+\ket{\bar{u}-e^{(j)}}\bra{\bar{v}-e^{(j)}})+o(\gamma^{2})]\\
&=&\frac{\sqrt{((1-\gamma)^{\|u\|^{2}}+(1-\gamma)^{\|\bar{u}\|^{2}})((1-\gamma)^{\|v\|^{2}}+(1-\gamma)^{\|\bar{v}\|^{2}})}}{2}\\
&&\times\ket{w_{1}}\bra{w_{2}}+\frac{\gamma}{2}(u\cdot v+u\cdot\bar{v}+\bar{u}\cdot v+\bar{u}\cdot\bar{v})\ket{w_{1}}\bra{w_{2}}+o(\gamma^{2})\\
&=&\frac{1}{2}(2-\frac{\gamma}{2}(\|u\|^{2}+\|\bar{u}\|^{2}+\|v\|^{2}+\|\bar{v}\|^{2}))\ket{w_{1}}\bra{w_{2}}\\
&&+\frac{\gamma}{2}(u+\bar{u})\cdot(v+\bar{v})\ket{w_{1}}\bra{w_{2}}+o(\gamma^{2})
\end{eqnarray*}
Since $\|u\|^{2}+\|\bar{u}\|^{2}=\|v\|^{2}+\|\bar{v}\|^{2}=(u+\bar{u})\cdot(v+\bar{v})=n$, the conclusion follows.
\end{proof}

\section{Codewords and Recovery}
It is clear from Theorem \ref{sc} that a self-complementary code corrects all the first-order errors in an amplitude-damping channel if and only if no confusion arises assuming the decay occurs at no more than one qubit. More precisely, $C=\mathrm{span}\{(\ket{u}+\ket{\bar{u}})/\sqrt{2}|u\in S\}$ is such a code if and only if the set $S\subset\{0,1\}^{n}$ satisfies
\begin{description}
\item{(S1)}If $u\in S$, then $\bar{u}\in S$; and
\item{(S2)}If $u, v\in S$ and $u-e^{(i)}=v-e^{(j)}$ for some $i,j$, then $u=v$.
\end{description}
The dimension of the code is $k=\mathrm{dim}C=|S|/2$. We use a greedy algorithm to search for maximal sets $S$ satisfied the above
conditions for some small values of $n$ and the results are listed in Table 1. The $\log_{2}k$ column approximately indicates the number of qubits that can be encoded. Linear regression yields that the slope of $y\sim x$ curve is 0.85, which is higher than 0.5 as in \cite{andrewthesis}.

\begin{center}
Table 1: Encoded Dimensions for Small Values of $n$

\begin{tabular}{ccc}
\\
\hline
$n$ & $k$ & $\log_{2}k$\\
\hline
4 & 2 & 1.00\\
5 & 2 & 1.00\\
6 & 5 & 2.32\\
7 & 8 & 3.00\\
8 & 12 & 3.58\\
9 & 18 & 4.17\\
10 & 41 & 5.36\\
11 & 78 & 6.28\\
12 & 146 & 7.19\\
13 & 273 & 8.09\\
14 & 515 & 9.01\\
15 & 931 & 9.86\\
16 & 1716 & 10.74\\
\hline
\end{tabular}
\end{center}

We now describe the algorithm to generate the recovery operation. We use the maximum likelihood recovery. 
\begin{algo}\label{recovery}The error vectors refer to vectors in $\{0,1\}^{n}$, where $1$ indicates a decay in that qubit. The error vectors are sorted by their weights and those with the same weight are sorted in the dictionary order.
\begin{description}
\item{Step 1} The first operator element of $\mathcal{R}$ corrects $\ket{f(u)}$ in \eqref{fu'} to $(\ket{u}+\ket{\bar{u}})/\sqrt{2}$. We are considering the error vector $e=(0,...,0)$.
\item{Step 2} Put $e$ to be its successor and come to construct the next operator element of $\mathcal{R}$.
\item{Step 3} Find a word $u\in S$ such as $u-e\geq0$, which means that the decay $e$ may happen to $u$, do the following:
\begin{description}
\item{Step 3.1} If $u-e\in S$ and both $u-e$ and $\overline{u-e}$ have not appeared in $\mathcal{R}$ except in the first operator element, then correct $\ket{g(u-e)}$ to $(\ket{u}+\ket{\bar{u}})/\sqrt{2}$. Go to Step 3.3.
\item{Step 3.2} If $u-e$ has not appeared in $\mathcal{R}$, then correct $\ket{u-e}$ to $(\ket{u}+\ket{\bar{u}})/\sqrt{2}$.
\item{Step 3.3} Search for the next word $u\in S$ with $u-e\geq0$ and go back to Step 3.1, until all the words in $S$ have been exhausted.
\end{description}
\item{Step 4} If the sum of the ranks of all the constructed operator elements of $\mathcal{R}$ equals $n$, then stop. Otherwise, go back to Step 2.
\end{description}
\end{algo}

By the construction, the operator elements of $\mathcal{R}$ are orthogonal projections, which makes it easy to implement as a quantum circuit. By Theorem \ref{sc}, $\mathcal{R}$ corrects all the first-order errors in the amplitude damping channel.

\section{An Example}
We demonstrate the performance of our codes by a simple example, the (8,12) code, which encodes a 12 dimensional space to 8 qubits. In this case
\begin{eqnarray*}
S&=&\{00000000,00000011,00001100,00110000.11000000,\\
&&10101000,01011000,01100100,10010100,11110000,\\
&&11001100,00111100 \text{ and their complements}\}
\end{eqnarray*}

The entanglement fidelity of this code on the ensemble $I/12$ is plotted and compared with the unprotected 3 qubits, since this code can encode at most 3 qubits. We can see that the code has the desired behavior when $\gamma\to 0$ (the linear term vanishes) and it also has a good performance for larger $\gamma$. See Fig. \ref{figure1}.

\begin{figure}
\centering
\includegraphics[width=8cm]{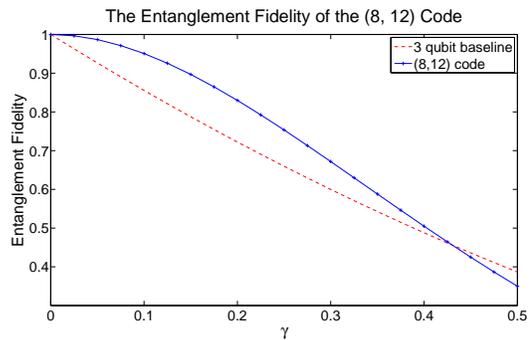}
\caption{The entanglement fidelity of the (8, 12) code versus the parameter $\gamma$, on the initial state $\rho_{0}=I/12$. The dashed curve is the entanglement fidelity of the bare 3 qubits on the initial state $I/8$.}\label{figure1}
\end{figure}

\section{Conclusion}
We have described the structure of a family of nonadditive quantum error correcting codes adapted to the amplitude damping channel with the emphasis on the significance of its self-complementary structure. The code has an extremely high rate and thus may turn out useful in quantum computation and quantum teleportation. Since the code is nonadditive, its decoding does not have the syndrome-diagnosis and recovery structure of the stabilizer codes, so we are concerned with designing the efficient quantum decoding circuits for these codes. It is also challenging to figure out the cardinality of $S$ satisfying (S1) and (S2) mentioned in the third section in a closed form or its asymptotic behavior. Table 1 is constructed using the greedy algorithm, so $S$ may well have a larger size than listed there.

This research is supported by MIT Undergraduate Research Opportunity Program.  PWS was supported in part by 
the W. M. Peck Center for Extreme Quantum Information Theory and by the National Science Foundation through
grant CCF-0431787.  The authors also thank Dr. Andrew Fletcher for his altruistic help.

\end{document}